\documentclass[11pt]{article}

\usepackage{authblk}

\usepackage{amsmath}
\usepackage{amsthm}
\usepackage{amssymb}
\usepackage{mathtools}
\usepackage{enumitem}

\usepackage{tikz}
\usetikzlibrary{calc, positioning, arrows, graphs, shapes, arrows.meta, decorations, decorations.markings}

\usepackage[margin=1in]{geometry}
\usepackage{thm-restate}

\usepackage{color}
\usepackage[pdfstartview=FitH,pdfpagemode=UseNone,colorlinks=true,citecolor=blue,linkcolor=blue]{hyperref}

\newtheorem{theorem}{Theorem}[section]
\newtheorem*{theorem*}{Theorem}
\newtheorem{claim}[theorem]{Claim}
\newtheorem{lemma}[theorem]{Lemma}

\newtheorem{remark}{Remark}
\newtheorem{construction}{Construction}
\newtheorem{definition}{Definition}
\newtheorem{corollary}{Corollary}
\newtheorem*{corollary*}{Corollary}
\newtheorem{conjecture}{Conjecture}

\newtheorem{fact}[conjecture]{Fact}
\newtheorem{problem}{Open Problem}

\DeclarePairedDelimiter\abs{\lvert}{\rvert}
\DeclarePairedDelimiter\norm{\lVert}{\rVert}

\DeclarePairedDelimiter\floor{\lfloor}{\rfloor}

\DeclarePairedDelimiter\ip{\langle}{\rangle}

\newcommand{\acc}{\cclass{{ACC}}^0}
\newcommand{\acz}{\cclass{{AC}}^0}
\newcommand{\AND}{\cclass{{AND}}}

\newcommand{\NOT}{\cclass{{NOT}}}
\newcommand{\maj}{\cclass{{MAJORITY}}}
\newcommand{\majn}{\cclass{{MAJORITY}}_n}
\newcommand{\dw}{\Delta_w}

\newcommand{\modm}{\cclass{{MOD}}_m}
\newcommand{\modg}{\cclass{{MOD}}}
\newcommand{\mdp}{\cclass{{MidBit}}^+}
\newcommand{\mdg}{\cclass{{MidBit}}}

\newcommand{\nulls}{\mathop{{nullspace}}}

\newcommand{\bspan}{\mathop{{row\text{--}span}}}

\newcommand{\tdeg}{\overline{{\mathop{{\deg}}}}}

\newcommand{\cclass}[1]{\mathsf{#1}}

\makeatletter
\newcommand{\dotDelta}{{\vphantom{\Delta}\mathpalette\d@tD@lta\relax}}
\newcommand{\d@tD@lta}[2]{%
  \ooalign{\hidewidth$\m@th#1\mkern-1mu\cdot$\hidewidth\cr$\m@th#1\Delta$\cr}%
}
\makeatother

\newcommand{\Ms}{\widehat{{M}}}
\newcommand{\Bs}{\widehat{{B}}}

\title{Degree Lower Bounds for Torus Polynomials and MAJORITY vs ACC\textsuperscript{0}}

\author[1]{Vaibhav Krishan}
\author[2]{Sundar Vishwanathan}
\affil[1]{The Institute of Mathematical Sciences, Chennai, India}
\affil[2]{Indian Institute of Technology Bombay, Mumbai, India}

\date{}

\begin{document}

\begin{titlepage}
\maketitle

\begin{abstract}
  The class \(\acc\) consists of Boolean functions that can be computed by constant-depth circuits of polynomial size with \(\AND, \NOT\) and \(\modm\) gates, where \(m\) is a natural number.
  At the frontier of our understanding lies a widely believed conjecture asserting that \(\maj\) does not belong to \(\acc\).

  A few years ago, Bhrushundi, Hosseini, Lovett and Rao (ITCS 2019) introduced torus polynomial approximations as an approach towards this conjecture.
  Torus polynomials approximate Boolean functions when the fractional part of their value on Boolean points is close to half the value of the function.
  They reduced the conjecture that \(\maj \notin \acc\) to a conjecture concerning the non-existence of low degree torus polynomials that approximate \(\maj\).

  We reduce the non-existence problem further, to a statement about finding feasible solutions for an infinite family of linear programs.
  The main advantage of this statement is that it allows for incremental progress, which means finding feasible solutions for successively larger collections of these programs.
  As an immediate first step, we find feasible solutions for a large class of these linear programs, leaving only a finite set for further consideration.
  Our method is inspired by the \emph{method of dual polynomials}, which is used to study the approximate degree of Boolean functions.
  Using our method, we also propose a way to progress further.

  We prove several additional key results with the same method, which include:
  \begin{itemize}
    \item
      A lower bound on the degree of \emph{symmetric} torus polynomials that approximate the \(\AND\) function.
      As a consequence, we get a separation that symmetric torus polynomials are \emph{weaker} than their asymmetric counterparts.
    \item
      An error-degree trade-off for symmetric torus polynomials approximating the \(\maj\) function, strengthening the corresponding result of Bhrushundi, Hosseini, Lovett and Rao (ITCS 2019).
    \item
      The \emph{first} lower bounds against torus polynomials approximating \(\AND\), showcasing the power of the machinery we develop.
      This lower bound nearly matches the corresponding upper bound.
      Hence, we get an almost complete characterization of the torus polynomial approximation degree of \(\AND\).
    \item 
      Lower bounds against asymmetric torus polynomials approximating \(\maj\), or \(\AND\), in the very low error regime.
      This partially answers a question posed in Bhrushundi, Hosseini, Lovett and Rao (ITCS 2019) about error-reduction for torus polynomials.
  \end{itemize}
\end{abstract}
\thispagestyle{empty}
\end{titlepage}

\section{Introduction}
\label{sec:intro}
Proving that an explicit function is not contained in a complexity class is the prime focus of complexity theorists, and such questions make up some of the hardest problems in computer science.
We study such a question at the frontier of our knowledge about Boolean circuit complexity classes.
To state the question, we first need to define the class of Boolean circuits we consider.

Denote by \(\acc\) the class of constant-depth Boolean circuits of polynomial size comprising \(\AND, \NOT\), and \(\modm\) gates.
A \(\modm\) gate outputs \(1\) if and only if the count of \(1\)s in the input is divisible by \(m\).
Nearly 35 years ago, Barrington~\cite{barrington1989bounded} conjectured that \(\acc\) does not contain \(\maj\).
Here, \(\maj\) outputs \(1\) if and only if the number of \(1\)s is at least half of the total number of inputs.
This conjecture has remained unresolved since.

\begin{conjecture}[Barrington's conjecture~\cite{barrington1989bounded}]
  \label{conj:majlb}
  \(\maj \notin \acc\).
\end{conjecture}

The objective of the conjecture is to prove that a particular circuit class cannot compute a certain function.
In the literature, this task is referred to as proving \emph{lower bounds} against that class.
We outline a few major approaches that have led to Boolean circuit lower bounds in the past.

One of the approaches is based on ``simplification'', and the probabilistic method, for example: using \emph{random restrictions}.
This is a classical technique, developed for studying various complexity classes, such as \(\acz\), the class of constant-depth circuits comprising \(\AND\) and \(\NOT\) gates.
A classic example is the landmark result of H{\aa}stad~\cite{hastad1986almost}, who proved that \(\acz\) circuits simplify when a significant fraction of the input variables are assigned values randomly.
It is easy to see that the parity function does not undergo such simplification.
The author formalized this intuition to prove nearly optimal lower bounds against \(\acz\) circuits.
However, random restrictions do not seem useful for lower bounds against \(\acc\), as it contains the parity function.

Another approach is to use the \emph{satisfiability algorithms}--to--lower bounds connection, to prove lower bounds from high complexity classes such as \(\cclass{NEXP}\).
Williams~\cite{williams2014nonuniform} developed this approach in a groundbreaking work, and used it to prove that \(\cclass{NEXP}\) is not contained in \(\acc\).
Subsequent works~\cite{chen2018average,murray2019circuit,chen2019non,chen2020almost,chen2023new} have considerably strengthened this connection.
In particular, Murray and Williams~\cite{murray2019circuit} have proved that \(\cclass{NQP}\) is not contained in \(\acc\), improving the lower bound.
However, it is not clear how to use this method to prove that an easily computable function, such as \(\maj\), is not contained in \(\acc\).

This leads us toward another classical approach, based on the \emph{polynomial method}.
In this framework, researchers study various notions of representing Boolean functions using polynomials.
This framework is quite powerful, and has numerous applications across theoretical computer science, see Aaronson's survey~\cite{aaronson2008polynomial} for an interesting and insightful account.
We describe two notions of polynomial representations that have found uses in the study of Boolean circuits.

The first notion, called real polynomial approximation, uses polynomials over the reals to approximate Boolean functions pointwise.
Nisan and Szegedy~\cite{nisan1994degree}, in a seminal work, studied this model, and proved its connections with other natural computational models, such as \emph{decision trees}.
Their work provided considerable impetus to the study of real polynomial approximations, and firmly established their use in mainstream complexity theory.
Today, the study of real polynomial approximations is a subfield by itself, with well-developed techniques for lower bounds and upper bounds, and numerous applications.
See Bun and Thaler's survey~\cite{bun2022approximate} for a comprehensive discussion of real polynomial approximations.

The second notion of polynomial approximation, called \emph{probabilistic polynomials}, uses a distribution of polynomials over a finite field.
On any Boolean point, a polynomial from the distribution should match the value of the given function with ``high'' probability.
Razborov~\cite{razborov1987lower}, and Smolensky~\cite{smolensky1987algebraic}, pioneered the use of this model in independent works.
Consider any function \(f\) computable by constant-depth circuits of polynomial size with \(\AND, \NOT\) and \(\modg_p\) gates, for a prime \(p\).
The authors, independently, proved that there exists a distribution of ``low'' degree polynomial over \(\mathbb{F}_p\) that matches \(f\) with ``high'' probability.
They also proved that the same does not hold for \(\maj\), or \(\modg_q\) for a prime \(q \neq p\), leading to a lower bound.

Broadly speaking, in order to prove that a function \(f\) is not contained in a class \(\mathcal{C}\), the theme here is to find a \emph{distinguisher}.
A distinguisher is a function \(\mu\) that maps \(f\) to a point outside the image of \(\mathcal{C}\) under \(\mu\), proving \(f \notin \mathcal{C}\).
That is, proving \(\mu(f) \notin \mu(\mathcal{C})\) implies \(f \notin \mathcal{C}\).
For example, in~\cite{razborov1987lower, smolensky1987algebraic}, the degree of the probabilistic polynomial acts as the distinguisher.
This degree is low for functions in \(\acz[p]\), while \(\maj\) requires large degree, hence the lower bound \(\maj \notin \acz[p]\).

The models of polynomial approximation defined above do not seem to give a distinguisher against \(\acc\).
For example, \(\acc\) circuits can use \(\modg_6\), which requires large degree in either model.
In fact, for a long time, there were no polynomial method based approaches known for proving \(\acc\) lower bounds.

Then, in a pivotal work few years ago, Bhrushundi, Hosseini, Lovett and Rao~\cite{bhrushundi2019torus} made an inspired suggestion of using the degree of torus polynomials as a distinguisher towards the \(\maj\) vs \(\acc\) question.
Torus polynomials approximate a Boolean function \(f\) if their fractional part is close to \(\frac{f}{2}\)\footnote{\(0\) and \(1\) have the same fractional part. Dividing \(f : \{0, 1\}^n \to \{0, 1\}\) by \(2\) makes the fractional parts different.}.
They proved that torus polynomial approximations extend both real polynomial approximations and approximation using probabilistic polynomials (see~\cite[Lemma 14]{bhrushundi2019torus}).
In fact, torus polynomials are more powerful, they can efficiently approximate \(\modm\) for any \(m\).
We define this model below, rephrasing~\cite[Definition 1]{bhrushundi2019torus}.

\begin{definition}[Torus Polynomial Approximation~\cite{bhrushundi2019torus}]
  \label{def:torus}
  Consider a Boolean function \(f : \{0, 1\}^n \to \{0, 1\}\), a polynomial \(P \in \mathbb{R}[X_1, \ldots, X_n]\) (we assume that \(P\) is multilinear without losing generality) and a real number \(0 \leq \varepsilon < \frac{1}{4}\).
  \(P\) is a torus polynomial that \(\varepsilon\)-approximates \(f\), if the following holds:

  There exists a function \(Z : \{0, 1\}^n \to \mathbb{Z}\), such that for any Boolean point \(a \in \{0, 1\}^n\), we have \(\abs*{P(a) - \frac{f(a)}{2} - Z(a)} \leq \varepsilon\).
  In other words, the \emph{fractional part} of \(P(a)\) is within \(\varepsilon\) of \(\frac{f(a)}{2}\).

  Denote the minimum degree of such a polynomial by \(\tdeg_\varepsilon(f)\).
\end{definition}

Given any function \(f : \{0, 1\}^n \to \{0, 1\}\), consider the unique multilinear polynomial that exactly matches \(f\) on all Boolean points\footnote{The existence of such a polynomial is folklore.}.
This would require degree \(n\) for most functions, with zero error of approximation.
Naturally, the question is, for which functions \(f\) can we construct a torus polynomial of much smaller degree, say \(\log^2(n)\), that \(\frac{1}{n^2}\)-approximate \(f\), for example.
In~\cite[Corollary 20]{bhrushundi2019torus}, the authors proved that something similar holds for all functions in \(\acc\).

\begin{restatable}[\cite{bhrushundi2019torus}]{theorem}{bhrushacc}
  \label{thm:bhrushacc}
  Consider any function \(f: \{0, 1\}^n \to \{0, 1\}\) such that \(f \in \acc\).
  Then, \(\tdeg_\varepsilon(f) \leq \log(n)^{O(1)}\) for any \(\varepsilon \geq \frac{1}{n^{O(1)}}\).
\end{restatable}

Hence, proving that the same does not hold for the majority function would prove \(\maj\) is not contained in \(\acc\).
This is precisely the goal of our work in this paper.
While we could not take this task to completion, we do make progress towards it.
Before discussing our contributions, we discuss previous work related to \(\acc\) as well as torus polynomials.

\subsection{Previous Work}

Various works in the 90's had proved conversion results for \(\acc\), with the purpose of proving \(\acc\) lower bounds.
For example, strengthening the results from~\cite{yao1990acc,beigel1994acc}, Green, K\"{o}bler, Regan, Schwentick and Tor{\'a}n~\cite{green1995power} proved that all \(\acc\) functions have \(\mdp\) circuits computing them.
Here, \(\mdp\) is the class of depth-two circuits, with \(\AND\) gates at the bottom and a \(\mdg\) gate\footnote{A \(\mdg\) outputs the middle bit from the binary expansion of the number of \(1\)s, with a fixed tie-breaking choice.} at the top.
The authors proposed to prove lower bounds against \(\mdp\) as an approach to prove \(\acc\) lower bounds.
They argued that the simpler structure of \(\mdp\) circuits might make it easier to prove lower bounds.

Along similar lines, by combining~\cite{haastad1991power,razborov1993nomega}, one gets a communication complexity based approach.
These works together imply that lower bounds against the \emph{number-on-forehead} communication model lead to \(\acc\) lower bounds.
However, lower bounds against this communication model also imply \(\mdp\) lower bounds.
Hence, logically speaking, lower bounds against \(\mdp\) are an easier route to \(\acc\) lower bounds.

In Krishan~\cite{krishan2021upper}, the author proved that functions with low degree torus polynomial approximations belong to \(\mdp\).
Hence, one can argue that lower bounds against torus polynomials is an even more refined approach for \(\acc\) lower bounds.
Moreover, Chen, Lu, Lyu and Oliveira~\cite{chen2021majority} proved that lower bounds against torus polynomials lead to \emph{average-case} lower bounds, and \emph{pseudorandom generators}, against \(\acc\).
Both of these are major open questions, which gives us further impetus for proving lower bounds against torus polynomials.

In~\cite{bhrushundi2019torus}, the authors introduced torus polynomials with the goal of proving that \(\maj\) does not belong to \(\acc\).
Now, \(\maj\) is a symmetric function, as its value only depends on the number of \(1\)s in the input.
Hence, it is natural to study \emph{symmetric}\footnote{A polynomial is symmetric if it is invariant under permutation of its variables.} torus polynomials that approximate \(\maj\) as the first step.
The authors used a counting-based argument to prove the following lower bound.

\begin{restatable}[Corollary 23 of~\cite{bhrushundi2019torus}]{theorem}{bhrushlb}
  \label{thm:bhrushlb}
  Any symmetric torus polynomial, that \(\frac{1}{20n}\)-approximates \(\maj\), must have degree \(\Omega\left(\sqrt{\frac{n}{\log(n)}}\right)\).
\end{restatable}

Now, a priori, this does not resolve Barrington's conjecture, it needs an analogous statement for \emph{asymmetric} torus polynomials.
Indeed, in~\cite[Conjecture 5]{bhrushundi2019torus}, the authors conjectured that Theorem~\ref{thm:bhrushlb} holds for asymmetric torus polynomials.

\begin{conjecture}[\cite{bhrushundi2019torus}]
  \label{conj:asymlb}
  Any torus polynomial, that \(\frac{1}{20n}\)-approximates \(\maj\), must have degree \(\Omega\left(\sqrt{\frac{n}{\log(n)}}\right)\).
\end{conjecture}

This conjecture, if true, proves Barrington's conjecture.
Moreover, it will separate \(\cclass{P}\) from \(\acc\), improving our knowledge well beyond what is currently known.

\subsection{Our Results}

The main goal of our work is to resolve Conjecture~\ref{conj:asymlb}.
Towards this end, we outline a plan of attack in Section~\ref{sec:majplan}.
A major contribution of this paper is a reduction of Conjecture~\ref{conj:asymlb} to a statement that allows for incremental progress.

Fix \(n\) and \(d = o\left(\sqrt{\frac{n}{\log(n)}}\right)\), and suppose the goal is to prove the following: Any torus polynomial that \(\varepsilon\)-approximates \(\maj\) over \(n\) variables must have degree more than \(d\).
Informally stated, we define a vector space \(\Gamma\) of dimension \(2^n\), a vector \(v'\) of dimension \(2^n\), and conjecture that for any vector \(v \in \mathbb{Z}^{2^n}\), there is a \(\gamma \in \Gamma\) such that:
\[
  \sum_{i=1}^{2^n} (v_i + v'_i) \gamma_i > \varepsilon \sum_{i=1}^{2^n} \abs{\gamma_i}
\]
The vector \(v'\) has a very simple structure, it encodes the function for which we wish to prove the lower bound, say \(\maj\).
See Theorem~\ref{thm:asymconvert} for an exact statement.
Incremental progress would mean proving the statement for larger and larger subsets of \(\mathbb{Z}^{2^n}\).

We take the first steps in this direction by first bounding the entries in \(v\), see Theorem~\ref{thm:zbounds}.
This immediately leaves us with only a finite subset of \(\mathbb{Z}^{2^n}\) for which we need to argue further.
Here, each entry of \(v\) can independently take a value that satisfies the upper bound, allowing \(2^n\) degrees of freedom.
As the next step, in Theorem~\ref{thm:inferz}, we show that a few entries in \(v\) determine the other entries.
This reduces the degrees of freedom to a tiny fraction of \(2^n\).

We demonstrate the power of our method by proving a lower bound against torus polynomials approximating \(\AND\).
No lower bounds were known for asymmetric torus polynomials prior to our work, and the lower bound we prove is only quadratically away from the corresponding upper bound for inverse-polynomial error.
Following is the formal statement of the result.

\begin{restatable}{theorem}{andlb}
  \label{thm:andlb}
  Any torus polynomial, that \(\varepsilon\)-approximates \(\AND\), must have degree \(\Omega\left(\log\left(\frac{1}{\varepsilon}\right)\right)\).
\end{restatable}

The result above also holds for \(\maj\), but it does not suffice to prove \(\maj \notin \acc\), which requires a lower bound of the form \(\log^{\omega(1)}(n)\) for some inverse-polynomial error.
Further, as a special case of the result above, we study the case when \(\varepsilon\) is exponentially small.
We show that \(\AND\) requires full degree in this regime, and so does \(\maj\).
In~\cite[Problem 6]{bhrushundi2019torus}, the authors ask about error-reduction for torus polynomials.
The full-degree lower bound allows us to conditionally answer this question for a particular error-regime, which we discuss in Subsection~\ref{subsec:tinyerr}.

Our method also applies to symmetric torus polynomials, which we use to prove a much stronger lower bound against symmetric torus polynomials approximating \(\AND\).
This lower bound matches the lower bound from~\cite[Corollary 23]{bhrushundi2019torus} (refer to Theorem~\ref{thm:bhrushlb}), but for \(\AND\), rather than \(\maj\).
Note that no such lower bound was known for \(\AND\) before our work, and it matches the corresponding upper bound, barring some log factors.
We state the result below.

\begin{restatable}{theorem}{andsymlb}
  \label{thm:andsymlb}
  Any symmetric torus polynomial, that \(\frac{1}{20n}\)-approximates \(\AND\), must have degree \(\Omega\left(\sqrt{\frac{n}{\log(n)}}\right)\).
\end{restatable}

The lower bound for symmetric torus polynomials is higher than the upper bound for asymmetric torus polynomials approximating \(\AND\).
Hence, as a corollary, we prove that symmetric torus polynomials are weaker than their asymmetric counterparts, see Corollary~\ref{corr:symsep}.
Moreover, this shows that a symmetrization based approach is unlikely to work for Conjecture~\ref{conj:asymlb}.

We also strengthen~\cite[Corollary 23]{bhrushundi2019torus}, and prove stronger degree lower bounds for smaller error.
If one follows the proof of~\cite[Corollary 23]{bhrushundi2019torus}, for symmetric torus polynomials that \(\frac{1}{n^2}\)-approximate \(\maj\), the lower bound remains the same, i.e. \(\Omega\left(\sqrt{\frac{n}{\log(n)}}\right)\).
We are able to prove \(\Omega\left(\sqrt{n}\right)\) as the lower bound in this case, strictly improving the degree, albeit by a log factor.
In fact, we are able to prove an error-degree trade-off, see Theorem~\ref{thm:symtradeoff} for the exact statement.

\subsubsection*{Our Method}

We use a linear programming based approach to prove our lower bound results.
This approach, based on duality in linear programs, allows us to find a witness that certifies the non-existence of a torus polynomial approximation.
We describe a broad outline of the method below.

Consider a torus polynomial \(P\), of degree at most \(d\), that \(\varepsilon\)-approximates \(f\).
Then, there exists some integer function \(Z : \{0, 1\}^n \to \mathbb{Z}\), such that \(P(a)\) is at most \(\varepsilon\) away from \(Z(a) + \frac{f(a)}{2}\) for any \(a \in \{0, 1\}^n\).
For each \(a\), we can write \(P(a)\) as a particular linear combination of its coefficients, hence, the condition above is a linear constraint.
Here, we make a crucial choice, with respect to how we treat \(Z\).
If we treat each \(Z(a)\) as an integer-valued variable, we get a system of linear Diophantine inequalities, that are much harder to handle.
Instead, we treat each \(Z(a)\) as an indeterminate, and write one linear program for each possible function \(Z\).

Thereby, we get an infinite family of linear programs, such that \(P\) exists if and only if some program from the family is feasible.
Hence, to prove that \(P\) does not exist, we need to prove that each program is infeasible, for which we look at the dual of these programs.
Using strong duality in linear programs, \(P\) does not exist if and only if each of the dual programs is feasible.

Now, given a function \(Z\), the dual program we obtain is as follows:
For each \(n\), and each degree \(d\), we have a matrix \(M(n, d)\), comprising the evaluations of all monomials of degree at most \(d\) on each Boolean point.
The set of solutions for the dual program consists of the nullspace of \(M(n, d)\), i.e., vectors \(\gamma\) such that \(M(n, d) \gamma = 0\).
Here, \(\gamma\) is a vector with \(2^n\) many real entries, with each entry \(\gamma_a\) indexed by a Boolean point \(a\).
A vector \(\gamma\) is a feasible solution for the dual, if it satisfies \(\abs*{\sum_{a \in \{0, 1\}^n} \gamma_a \left(Z(a) + \frac{f(a)}{2}\right)} > \varepsilon \sum_{a \in \{0, 1\}^n} \abs{\gamma_a}\).
For a detailed explanation on how we obtain the dual, please refer to Theorem~\ref{thm:asymconvert}.

With this method, our plan to prove Conjecture~\ref{conj:asymlb} is as follows.
Fix \(n\), and any \(d = o\left(\sqrt{\frac{n}{\log(n)}}\right)\), with \(f = \maj\) and \(\varepsilon = \frac{1}{20n}\).
For any function \(Z : \{0, 1\}^n \to \mathbb{Z}\), we plan to find a feasible solution \(\gamma\) for the dual corresponding to \(Z\).
To start, we find such a feasible solution if \(\abs{Z(a)}\) exceeds an upper bound for any \(a \in \{0, 1\}^n\), where the upper bound depends on \(a\).
Then, we show how to infer the values of \(Z(a)\) for each \(a\) with Hamming weight \(\abs{a} \geq (d+1)^2\), when all \(Z(a)\) for \(\abs{a} < (d+1)^2\) are fixed.
In other words, if \(Z(a)\) does not take the inferred value for some \(a\), we find a feasible solution for the dual corresponding to \(Z\).
We propose to continue this plan further, by finding more feasible solutions to incrementally rule out all possible \(Z\)s.

Our method shares some similarities to how dual polynomials are used to prove lower bounds for real polynomial approximations.
Each vector \(\gamma\) in the nullspace of \(M(n, d)\) is in fact a dual polynomial, with \emph{pure high-degree} more than \(d\).
We use a geometric perspective, as we find it more useful, especially for our lower bound results.

\subsubsection*{Organization}
We discuss some preliminaries, and a general method for proving torus polynomial lower bounds, in Section~\ref{sec:prelims}.
Based on this method, we propose a plan to prove Conjecture~\ref{conj:asymlb} in Section~\ref{sec:majplan}.
Our lower bound results for asymmetric torus polynomials appear in Section~\ref{sec:lb}.
Finally, we prove lower bounds for symmetric torus polynomials in Section~\ref{sec:symlb}.

\section{Preliminaries}
\label{sec:prelims}
We consider natural numbers without including \(0\), and denote it by \(\mathbb{N} = \{1, 2, \ldots\}\).
For \(n \in \mathbb{N}\) and \(d \in \mathbb{N}\), we use the following notation for brevity:
\begin{align*}
  \binom{n}{\leq d} &= \sum_{i \leq d} \binom{n}{i} &
  \binom{n}{> d} &= \sum_{i > d} \binom{n}{i} \\
  [n] &= \{1, \ldots, n\} &
  [n]^* &= \{0, \ldots, n\} \\
  \binom{[n]}{\leq d} &= \{S : S \subseteq [n], \abs{S} \leq d\} &
  \binom{[n]}{> d} &= \{S : S \subseteq [n], \abs{S} > d\}
\end{align*}

\subsection{Sets and Boolean Points}

We identify \(2^{[n]}\) with \(\{0, 1\}^n\) in the natural way, as follows.
For \(S \subseteq [n]\), the corresponding Boolean point has a \(1\) at position \(i\) if and only if \(i\) is present in \(S\).
This defines a bijection between \(2^{[n]}\) and \(\{0, 1\}^n\).
Using this bijection, we will often interpret a set \(S \subseteq [n]\) as a Boolean point, and a Boolean point \(a \in \{0, 1\}^n\) as a set, making the interpretation explicit wherever it is not clear from the context.
\(\abs{a}\) denotes the Hamming weight of \(a \in \{0, 1\}^n\), which also equals its size when considered as a set.

\subsection{Linear Algebra}

Over the reals \(\mathbb{R}\), we denote the set of matrices of size \(m \times n\) by \(\mathcal{M}_{m \times n}(\mathbb{R})\).
For a matrix \(M \in \mathcal{M}_{m \times n}(\mathbb{R})\), we denote its nullspace by \(\nulls(M) = \{\gamma \in \mathbb{R}^n : M \gamma = 0\}\).
For a vector \(\gamma \in \mathbb{R}^n\), we denote its \(\ell^p\)-norm by \(\norm{\gamma}_p = (\sum_{i=1}^n \abs{\gamma_i}^p)^{\frac{1}{p}}\).
Of particular interest to us are the \(\ell^1\)-norm and the \(\ell^2\)-norm, defined for \(p = 1, 2\) respectively.
We will also consider the \(\ell^\infty\)-norm, defined as \(\norm{\gamma}_\infty = \max_{i=1}^n \abs{\gamma_i}\).

\subsection{Our Method for Torus Polynomial Lower Bounds}

Now, we describe a general method for proving lower bounds on torus polynomial approximations.
Fix \(n \in \mathbb{N}\) and \(d \in \mathbb{N}\), such that \(d < n\).
Also, fix a Boolean function \(f : \{0, 1\}^n \to \{0, 1\}\), and an error of approximation \(\varepsilon < \frac{1}{4}\).
We start by defining a family of set-inclusion matrices that will be relevant for the method.
For each monomial of degree at most \(d\) over \(n\) variables, the matrix contains one row, and the row encodes the evaluation of the monomial over all Boolean points.

\begin{construction}
  \label{cons:matm}
  Define the matrix \(M(n, d)\) of size \(\binom{n}{\leq d} \times 2^n\) as follows.
  Its rows are indexed by elements of \(\binom{[n]}{\leq d}\), and columns by elements of \(2^{[n]}\).
  The entries for \(1 \leq i \leq \binom{n}{\leq d}, 1 \leq j \leq 2^n\) are:
  \[
    M_{i, j} = 1_{S_i \subseteq S_j}
  \]
  Here, \(1_{S_i \subseteq S_j}\) is an indicator function, evaluating to \(1\) if \(S_i \subseteq S_j\), and zero otherwise.
\end{construction}

In the following statement, we convert the question of torus polynomial lower bounds to an existence based question.
The proof of the statement is very similar to that for the method of dual polynomials (see~\cite{bun2022approximate}), we present it here for the sake of completeness.

\begin{theorem}
  \label{thm:asymconvert}
  The following are equivalent for any \(n \in \mathbb{N}, d \in \mathbb{N}\) such that \(d < n\), \(\varepsilon < \frac{1}{4}\) and \(f : \{0, 1\}^n \to \{0, 1\}\):

  \begin{itemize}
    \item Any torus polynomial that \(\varepsilon\)-approximates \(f\) has degree more than \(d\).
    \item For any \(Z : \{0, 1\}^n \to \mathbb{Z}\), there exists a vector \(\gamma \in \nulls(M(n, d))\), such that:
      \[\abs*{\ip*{Z + \frac{f}{2}, \gamma}} > \varepsilon \norm{\gamma}_1\]
  \end{itemize}
\end{theorem}

\begin{proof}
  To start the proof, assume that there exists a torus polynomial \(P\) of degree at most \(d\), that \(\varepsilon\)-approximates \(f\).
  As per Definition~\ref{def:torus}, there exists \(Z : \{0, 1\}^n \to \mathbb{Z}\), such that \(P(a)\) is within \(\varepsilon\) distance from \(Z(a) + \frac{f(a)}{2}\) for each Boolean point \(a\).
  We plan to use linear programming to capture the conditions that \(P\) must satisfy.
  Denote the coefficients of \(P\) as a vector \(\alpha \in \mathbb{R}^{\binom{n}{\leq d}}\), where each entry of \(\alpha\), indexed by a set \(S\) of size at most \(d\), acts as a variable in the program.
  For \(a \in \{0, 1\}^n\), the value of \(P(a) = \sum_{S \in \binom{[n]}{\leq d}} \alpha_S \cdot 1_{S \subseteq a}\).
  Note that we can rewrite this expression as the inner-product of \(\alpha\) with the column of \(M(n, d)\) indexed by \(a\).
  We treat \(Z\) as an indeterminate, and collect the constraints over all \(a \in \{0, 1\}^n\) as the following linear program:

  \begin{center}
    \begin{tabular}{c c c}
      \(\min_{\alpha}\) & \(\epsilon\) \\
      such that & \(\abs*{\alpha^T M(n, d) - Z - \frac{f}{2}} \leq \epsilon\) \\
                & \(\alpha \in \mathbb{R}^{\binom{n}{\leq d}}\)
    \end{tabular}
  \end{center}

  Now, \(P\) is a torus polynomial that \(\varepsilon\)-approximates \(f\) if and only if the program above achieves an objective of \(\varepsilon\) for some \(Z\).
  Hence, our goal is to prove that all the programs above achieve an optimal objective value more than \(\varepsilon\).
  We use standard tools from the theory of linear programming, and write the dual of these programs as follows.

  \begin{center}
    \begin{tabular}{c c c}
      \(\max_{\gamma}\) & \(\ip*{\gamma, Z + \frac{f}{2}}\) & \\
      such that & \(\gamma \in \nulls(M(n, d))\) & \\
                & \(\norm{\gamma}_1 = 1\) & 
    \end{tabular}
  \end{center}

  Using strong duality, the non-existence of \(P\) is equivalent to the statement that each dual above achieves an objective more than \(\varepsilon\).
  To finish the proof, note that the conditions \(\ip*{\gamma, Z + \frac{f}{2}} > \varepsilon\) and \(\norm{\gamma}_1 = 1\) together are equivalent to the condition \(\ip*{\gamma, Z + \frac{f}{2}} > \varepsilon \norm{\gamma}_1\).
  Finally, if \(\abs*{\ip*{\gamma, Z + \frac{f}{2}}} > \varepsilon \norm{\gamma}_1\), then either \(\ip*{\gamma, Z + \frac{f}{2}} > \varepsilon \norm{\gamma}_1\) or \(\ip*{-\gamma, Z + \frac{f}{2}} > \varepsilon \norm{-\gamma}_1\) holds, which suffices for our purpose.
  This completes the proof.
\end{proof}

\section{Plan for Majority}
\label{sec:majplan}
In this section, we set up a program towards proving Conjecture~\ref{conj:asymlb}.
Fix the function for this section \(f = \maj\), some large enough \(n \in \mathbb{N}\), any \(d \leq O\left(\sqrt{\frac{n}{\log(n)}}\right)\), and \(\varepsilon = \frac{1}{20n}\).
Call a function \(Z : \{0, 1\}^n \to \mathbb{Z}\) \emph{good}, if there is a witness \(\gamma \in \nulls(M(n, d))\) such that:
\[
  \abs*{\ip*{Z + \frac{\maj}{2}, \gamma}} > \varepsilon \norm*{\gamma}_1
\]
In other words, \(Z\) is good, if we can prove that the dual corresponding to \(Z\) is feasible.
Conjecture~\ref{conj:asymlb} posits that all \(Z\)s are good.
In this language, the lower bound result by Bhrushundi, Hosseini, Lovett and Rao~\cite[Corollary 3.3]{bhrushundi2019torus} states that all \emph{symmetric} \(Z\)s are good.

Now, the challenge is that we need to argue about infinitely many \(Z\)s, and find a feasible solution \(\gamma\) from the vector space \(\nulls(M(n, d))\), which is also infinite.
The latter is easy to fix using the theory of linear programming, as we can choose \(\gamma\) from the set of \emph{basic solutions}, which is a finite set.
Therefore, we focus on the set of good \(Z\)s, and show that it is cofinite.
In other words, we show that all \(Z\)s, except for a finite set, are good.

Towards this goal, we first make an observation that allows us to consider \(Z\)s as functions over a smaller domain.
Consider two functions \(Z\) and \(Z'\), such that \(Z - Z' \in \bspan(M)\).
Then, \(\ip{Z, \gamma}\) equals \(\ip{Z', \gamma}\) for any \(\gamma \in \nulls(M)\).
Hence, if \(Z\) is good, then any \(Z'\) such that \(Z - Z' \in \bspan(M)\) is also good.
Therefore, define the relation \(Z \sim Z'\) if \(Z - Z' \in \bspan(M)\), which is an equivalence relation.
We can pick a representative \(Z\) from each equivalence class of the relation, and it suffices to prove that the representative \(Z\) is good.
The following statement formalizes this intuition.

\begin{lemma}
  \label{lem:zerozs}
  Consider a polynomial \(P \in \mathbb{R}[X_1, \ldots, X_n]\) of degree at most \(d\).
  Then, there exists a polynomial \(P' \in \mathbb{R}[X_1, \ldots, X_n]\) of degree at most \(d\), satisfying the following conditions:
  \begin{itemize}
    \item If \(P\) \(\varepsilon\)-approximates \(f\), then \(P'\) also \(\varepsilon\)-approximates \(f\).

    \item If \(Z'\) denotes the integer part of \(P'\), then \(Z'(a) = 0\) for any \(a\) with \(\abs{a} \leq d\).
  \end{itemize}
\end{lemma}

\begin{proof}
  The proof follows from an inductive procedure, where we start with \(\abs{a} = 0\), and end after considering \(\abs{a} = d\).
  The initial polynomial \(P\) is of the form
  \[P = \sum_{S \in \binom{[n]}{\leq d}} c_S X^S\]
  As we perform our inductive procedure, after considering the case of \(\abs{a} = i\), we will obtain a polynomial \(P_i\), which will satisfy the following conditions:
  \begin{itemize}
    \item \(P_i\) has degree at most \(d\).
    \item For any \(a \in \{0, 1\}^n\), \(P(a) - P_i(a)\) is an integer.
    \item If \(P\) is a torus polynomial that \(\varepsilon\)-approximates \(f\), then, for any \(a \in \{0, 1\}^n\) with \(\abs{a} \leq i\),
      \[P_i(a) \in \left[-\varepsilon + \frac{f(a)}{2}, \varepsilon + \frac{f(a)}{2}\right].\]
  \end{itemize}

  For a point \(a \in \{0, 1\}^n\), \(P(a) = \sum_{S \subseteq a} c_S\).
  The crucial observation for this inductive procedure is that, if \(\abs{a} \leq i\), then, the expression for \(P(a)\) only depends on \(S\) such that \(\abs{S} \leq i\).
  Therefore, any change to a coefficient \(c_S\) with \(\abs{S} = i\) does not affect \(P(a)\) with \(\abs{a} < i\).

  For the base case, consider the all zeros point, \(a = 0\).
  Modify \(P\) by subtracting \(m_0 = \left\lfloor P(0) - \frac{f(0)}{2} \right\rceil\) from it, to get \(P_0\).
  Note that \(P_0\) is as desired.

  For the inductive step, say \(P_i\) satisfies all the three conditions listed above.
  For each point \(a\) with \(\abs{a} = i+1\), subtract \(m_a = \left\lfloor P_i(a) - \frac{f(a)}{2} \right\rceil \prod_{i \in a} X_i\) from \(P_i\).
  After subtracting each \(m_a\), denote the obtained polynomial as \(P_{i+1}\).
  One can verify easily that \(P_{i+1}\) satisfies the required conditions for the inductive procedure.

  The final polynomial \(P' = P_d\) has the desired properties to complete the proof.
\end{proof}

Henceforth, we will always assume that \(Z(a) = 0\) for any \(a \in \{0, 1\}^n\) with \(\abs{a} \leq d\).
Now, we begin the task of proving that certain \(Z\)s are good.
We will produce the vectors we need for this task using the following construction, which is only a minor generalization of the construction from~\cite{bun2015dual}.

\begin{construction}
  \label{cons:genray}
  Construct a vector \(\gamma\) as follows.

  \begin{itemize}
    \item[] \textbf{Input}: 
      
      \begin{itemize}
        \item two natural numbers \(n \in \mathbb{N}\) and \(d \in [n-1]^*\),
        \item two subsets \(S_1, S_2 \subseteq [n]\), such that \(S_1 \subseteq S_2\), and \(\abs{S_2 \setminus S_1} \geq d+1\),
        \item a set \(I \subseteq [\abs{S_2 \setminus S_1}]^*\) of size \(\abs{I} = d+2\).
      \end{itemize}
    \item[] \textbf{Output}: A vector \(\gamma \in \mathbb{R}^{2^n}\).
    \item[] \textbf{Construction}:
      First, define a univariate polynomial \(q_{I}(t) = \prod_{i \in [k]^* \setminus I} (t-i)\), where \(k = \abs{S_2 \setminus S_1}\).
      For any set \(T\), if \(S_1 \subseteq T \subseteq S_2\), keep \(\gamma_T = (-1)^{\abs{T}} q_{I}(\abs{T \setminus S_1})\).
      Otherwise, keep \(\gamma_T = 0\).

      Output \(\gamma\).
  \end{itemize}
\end{construction}

We claim that the construction produces a vector in \(\nulls(M(n, d))\).
The proof is similar to the proof of~\cite[Lemma 31]{bun2022approximate}, we present it here for the sake of completeness.

\begin{lemma}
  \label{lem:genray}
  Consider any \(n \in \mathbb{N}\), \(d \in [n-1]^*\), \(S_1, S_2 \subseteq [n]\) such that \(S_1 \subseteq S_2\) and \(\abs{S_2 \setminus S_1} \geq d+1\), and \(I \subseteq [\abs{S_2 \setminus S_1}]^*\) of size \(\abs{I} = d+2\).
  Then, Construction~\ref{cons:genray}, on input \(n, d, S_1, S_2\) and \(I\), outputs a vector \(\gamma \in \nulls(M(n, d))\).
\end{lemma}

\begin{proof}[Proof of Lemma~\ref{lem:genray}]
  We follow along the proof of~\cite[Lemma 31]{bun2022approximate}.
  For simplicity, consider the case when \(S_1 = \emptyset, S_2 = [n]\).
  Choose any set \(I \subseteq [n]^*\) of size \(\abs{I} = d+2\).
  Construct \(\gamma\) as described in the statement, using Construction~\ref{cons:genray}.
  We need to prove that \(M \gamma = 0\).

  Consider any row of \(M\) as a vector \(\textbf{m}\).
  We need to prove that \(\ip{\textbf{m}, \gamma} = 0\).
  Recall the construction for \(M\), from Construction~\ref{cons:matm}.
  Say \(\textbf{m}\) is the row of \(M\) indexed by \(S \subseteq [n]\).
  Then, \(\textbf{m}\) corresponds to the evaluations of the monomial \(\prod_{i \in S} x_i\) over all Boolean points.
  Note that \(\abs{S} \leq d\), hence any monomial we consider here has degree at most \(d\).
  We abuse notation, and use \(\textbf{m}\) to denote this monomial as well.

  To prove \(\ip{\textbf{m}, \gamma} = 0\), we expand this expression.
  \[
    \ip{\textbf{m}, \gamma} = \sum_{S \subseteq [n]} \textbf{m}(S) \gamma(S)
  \]
  Now, notice that \(\gamma\) is a symmetric vector, i.e. \(\gamma_{S_1} = \gamma_{S_2}\) if \(\abs{S_1} = \abs{S_2}\).
  Hence, we can rewrite the expression above as:
  \[
    \ip{\textbf{m}, \gamma} = \sum_{t \in [n]^*} \gamma_{[t]} \sum_{S \in \binom{[n]}{t}} m(S)
  \]
  Next, a straightforward calculation yields \(\sum_{S \in \binom{[n]}{t}} m(S) = \binom{n - \deg(m)}{t - \deg(m)}\).
  Moreover, \(\gamma_{[t]} = (-1)^t q_{I}(t)\) as per the construction of \(\gamma\).
  Therefore, we get the following expression:
  \[
    \ip{\textbf{m}, \gamma} = \sum_{t \in [n]^*} (-1)^t \binom{n-\deg(m)}{t-\deg(m)} q_{I}(t)
  \]
  The following commonly-known fact (cf.~\cite[Fact 30]{bun2022approximate}) allows us to complete the proof.
  \begin{fact}[Folklore]
    \label{fact:degzerosum}
    For any univariate polynomial \(p \in \mathbb{R}[t]\) of degree at most \(n-d-1\), and any \(d' \leq d\), the following holds:
    \[
      \sum_{t=0}^n (-1)^t \binom{n-d'}{t-d'} p(t) = 0
    \]
  \end{fact}
  To see why this fact is useful, just note that \(\deg(m) \leq d\), and \(q_I\) has degree \(n-d-1\).

  The general case of \(S_1, S_2, I\) follows similarly.
  We leave the proof to the reader.
\end{proof}

Using these vectors constructed above, we prove an upper bound on the value of each \(Z(a)\).
In other words, for each \(a\) we describe an upper bound, such that if \(Z(a)\) violates this bound, then \(Z\) is good.
The reader may think of this as an upper bound on an appropriately defined weighted \(\ell^\infty\)-norm of \(Z\).

\begin{theorem}
  \label{thm:zbounds}
  Choose a large enough \(n \in \mathbb{N}\).
  Consider a torus polynomial \(P\), of degree at most \(d < \frac{n}{2}\), that approximates \(f = \maj\) within an error of \(\varepsilon\).
  Then, the following holds for any \(a \in \{0, 1\}^n\) and the function \(Z\) corresponding to \(P\):
  \[\abs{Z(a)} \leq \varepsilon 2^{d+1} \binom{\abs{a}}{d+1}\]
  In other words, if \(\abs{Z(a)} > \varepsilon 2^{d+1} \binom{\abs{a}}{d+1}\) for some \(a \in \{0, 1\}^n\), then \(Z\) is \emph{good}.

  The statement above holds for any \(f\) such that \(f(a) = 0\) for each \(a \in \binom{[n]}{\leq d}\), e.g. \(f = \AND\)\footnote{As noted by an anonymous ITCS 2026 reviewer.}.
\end{theorem}

\begin{proof}
  Consider a set \(S\) of size \(\abs{S} = k \geq d+1\).
  Construct a ray \(\overline{\gamma} \in \mathop{null}(M)\) using Construction~\ref{cons:genray} with \(S_1 = \emptyset, S_2 = S\) and \(I = [d]^* \cup \{k\}\).
  The values we need to compare are \(\ip*{Z + \frac{f}{2}, \overline{\gamma}}\) and \(\varepsilon \norm{\overline{\gamma}}_1\).
  Note that \(Z(a) = f(a) = 0\) for any \(a\) with \(\abs{a} \leq d\).
  Hence, \(\ip*{\overline{\gamma}, Z+\frac{f}{2}} = \overline{\gamma}_S \left(Z(S) + \frac{f(S)}{2}\right)\).
  Therefore, if \(\abs*{Z(S) + \frac{f(S)}{2}} > \frac{\varepsilon \norm{\overline{\gamma}}_1}{\abs{\overline{\gamma}_S}}\), then \(Z\) is good.

  For a set \(T \subseteq S\) of size \(\abs{T} \in \{d+1, \ldots, k-1\}\), \(\overline{\gamma}_T = 0\) by construction.
  Then, \(\abs{\overline{\gamma}_S} = \prod_{i=1}^{k-(d+1)} i = (k-d-1)!\).
  Finally, for a set \(T \subseteq S\) of size \(\abs{T} = t \leq d\), \(\abs{\overline{\gamma}_T} = \prod_{i=d+1}^{k-1} (i-t) = \frac{(k-1-t)!}{(d-t)!}\).
  Note that there are \(\binom{k}{t}\) such sets \(T\).

  Hence, \(\norm{\overline{\gamma}}_1 = (k-d-1)! + \sum_{t=0}^d \binom{k}{t} \frac{(k-1-t)!}{(d-t)!}\).
  Dividing by \(\abs{\overline{\gamma}_S}\), we get:
  \begin{align*}
    \frac{\norm{\overline{\gamma}}_1}{\abs{\overline{\gamma}_S}} &= 1 + \sum_{t=0}^d \binom{k}{t} \frac{(k-1-t)!}{(d-t)!(k-d-1)!} &
                                                                 &= 1 + \sum_{t=0}^d \frac{k!}{t!(d-t)!(k-d-1)!(k-t)} & \\
                                           &= 1 + \binom{k}{d+1} \sum_{t=0}^d \binom{d}{t} \frac{d+1}{k-t} &
                                           &\leq 1 + \binom{k}{d+1} \sum_{t=0}^{d} \binom{d}{t} \frac{d+1}{d+1-t} & \\
                                           &= 1 + \binom{k}{d+1} \sum_{t=0}^d \binom{d+1}{t} &
                                           &= 1 + \binom{k}{d+1} \left(2^{d+1} - 1\right) & \\
                                           &\leq 2^{d+1} \binom{k}{d+1}
  \end{align*}
  This proves the claim as desired.
\end{proof}

With the previous result, we have only a finite set of \(Z\)s remaining, that we need to prove are good.
Now, we argue that we only need to look at a few values of \(Z(a)\) to figure out the rest of them.
This reduces the degrees of freedom for the family of linear programs we study, from \(2^n\) to a small fraction of \(2^n\).
To that effect, we introduce the following definition.

\begin{definition}[\(\varepsilon\)-isolator]
  A vector \(\gamma \in \nulls(M(n, d))\) is defined to be an \(\varepsilon\)-\emph{isolator} for a set \(a \subseteq [n]\) if the following holds:
  Given values of \(Z(b) \in \mathbb{Z}\) for every \(b \subsetneq a\), there is at most one value \(z_a \in \mathbb{Z}\) for \(Z(a)\) such that \(\abs*{\ip*{Z + \frac{f}{2}, \gamma}} \leq \varepsilon \norm{\gamma}_1\).

  In other words, if \(Z(a) \neq z_a\), then \(Z\) is good with \(\gamma\) as the witness.
\end{definition}

We show that for large enough \(n\) and any \(d < \sqrt{n}-1\), any point \(a\) with \(\abs{a} \geq (d+1)^2\) has an \(\varepsilon\)-isolator for \(\varepsilon = \frac{1}{20 n}\).
Note that if \(d = O(\sqrt{n})\) is small enough, then a large fraction of points \(a\) have an \(\varepsilon\)-isolator, hence, reducing the degrees of freedom.
The formal statement is as follows.

\begin{theorem}
  For any large enough \(n \in \mathbb{N}\), any \(d < \sqrt{n}-1\), and \(\varepsilon = \frac{1}{20 n}\), any \(a \in \{0, 1\}^n\) with \(\abs{a} \geq (d+1)^2\) has an \(\varepsilon\)-isolator.
\end{theorem}

The above statement follows from the statement below.

\begin{theorem}
  \label{thm:inferz}
  For any large enough \(n \in \mathbb{N}\), any \(d < \sqrt{n}-1\), and \(\varepsilon = \frac{1}{20 n}\), the following holds:
  Fix values of \(Z(b)\) for each point \(b \in \{0, 1\}^n \) with \(\abs{b} < (d+1)^2\).
  Then, consider a point \(a \in \{0, 1\}^n\) with \(\abs{a} \geq (d+1)^2\).
  For each \(a' \subseteq a\) with \(\abs{a'} = \abs{a} - (d+1)^2\), define the following rational number:
  \[
    R_{a, a'} = 2 \sum_{i=1}^{d+1} (-1)^i \frac{\binom{d+1}{i}}{\binom{d+i+1}{i}}
    \left(\frac{\sum_{a' \subseteq b \subseteq a: \abs{b} = \abs{a} - i^2} Z(b) + \frac{f(b)}{2}}{\binom{(d+1)^2}{i^2}}\right)
  \]
  If the following holds for any \(a'\):
  \[
    \abs*{Z(a) + R_{a, a'} + \frac{f(a)}{2}} > \frac{1}{\sqrt{n}}
  \]
  Then, \(Z\) is good.
  Note that all choices of \(Z(a) \in \mathbb{Z}\), except for possibly a single choice, lead to a good \(Z\).
  In other words, \(Z(a)\) is uniquely determined.
\end{theorem}

Before we begin the proof, we would like to explain the expression in simple words as it may look complicated at first glance.
First, the expression only looks at \(Z(b) + \frac{f(b)}{2}\) where \(b\) belongs to the Hamming sub-cube between \(a'\) and \(a\).
The expression \(\frac{\sum_{a' \subseteq b \subseteq a: \abs{b} = \abs{a} - i^2} Z(b) + \frac{f(b)}{2}}{\binom{(d+1)^2}{i^2}}\) is simply the average over the Hamming layer at distance \(i^2\) below \(a\).
This average is multiplied with the coefficient \(2 \cdot (-1)^i \frac{\binom{d+1}{i}}{\binom{d+i+1}{i}}\) to obtain the expression.
Now, we begin the proof.

\begin{proof}
  Fix a point \(a \in \{0, 1\}^n\) with \(\abs{a} \geq (d+1)^2\), and choose a subset \(a' \subseteq a\) with \(\abs{a'} = \abs{a} - (d+1)^2\).
  Use Construction~\ref{lem:genray} with \(n, d, S_1 = a', S_2 = a, I = \left\{(d+1)^2 - i^2 : i \in [d+1]^*\right\}\), to obtain \(\overline{\gamma} \in \nulls(M)\).
  We denote the polynomial \(q_{I}\) from the construction by \(q\) for brevity.

  Now, assume that the condition from the statement holds, which is as follows:
  \begin{equation}
    \abs*{Z(a) + R_{a, a'} + \frac{f(a)}{2}} > \frac{1}{\sqrt{n}}
    \label{ineq:zplusmore}
  \end{equation}

  Then, we claim that \(Z\) is good, with \(\overline{\gamma}\) as the witness, i.e.,

  \begin{equation}
    \abs*{\ip*{Z + \frac{f}{2}, \overline{\gamma}}} > \varepsilon \norm{\overline{\gamma}}_1
    \label{ineq:zipmore}
  \end{equation}

  To see why~\ref{ineq:zplusmore} implies~\ref{ineq:zipmore}, consider the following equalities:
  \begin{align}
    \ip*{Z + \frac{f}{2}, \overline{\gamma}} &= \sum_{b \in \{0, 1\}^n} \left(Z(b) + \frac{f(b)}{2}\right) \overline{\gamma}_b \label{eqn:expandip} \\
                                             &= \sum_{a' \subseteq b \subseteq a} \left(Z(b) + \frac{f(b)}{2}\right) \overline{\gamma}_b \label{eqn:supportgamma} \\
                                             &= \sum_{i=0}^{d+1} (-1)^{i^2} q((d+1)^2 - i^2) \sum_{a' \subseteq b \subseteq a : \abs{b} = \abs{a} - i^2} \left(Z(b) + \frac{f(b)}{2}\right) \label{eqn:gammasub}
  \end{align}
  Equality~\ref{eqn:expandip} follows by expanding the expression.
  Note that \(\overline{\gamma}_b = 0\) for any \(b\) with either \(a' \not\subseteq b\) or \(b \not\subseteq a\), as per the construction of \(\overline{\gamma}\).
  Hence, equality~\ref{eqn:supportgamma} follows.
  Finally, we substitute the values of \(\overline{\gamma}_b\) to obtain equality~\ref{eqn:gammasub}.

  Now, to calculate the RHS of inequality~\ref{ineq:zipmore}, expand the \(\ell^1\)-norm of \(\overline{\gamma}\) as \(\norm{\overline{\gamma}}_1 = \sum_{i=0}^{d+1} \binom{(d+1)^2}{i^2} \abs*{q((d+1)^2 - i^2)}\). 
  We divide the expression of the inner-product, as well as the \(\ell^1\)-norm, to modify inequality~\ref{ineq:zipmore} as:
  \begin{equation}
    \abs*{\sum_{i=0}^{d+1} \frac{q((d+1)^2 - i^2)}{q((d+1)^2)} \left(\sum_{\substack{a' \subseteq b \subseteq a \\ \abs{a} - \abs{b} = i^2}} Z(b) + \frac{f(b)}{2}\right)} > \varepsilon \sum_{i=0}^{d+1} \binom{(d+1)^2}{i^2} \frac{\abs*{q((d+1)^2 - i^2)}}{q((d+1)^2)} \label{ineq:expand}
  \end{equation}
  Note that \(q((d+1)^2) > 0\), hence, dividing by \(q((d+1)^2)\) does not change the direction of the inequality.
  Next, we calculate \(\binom{(d+1)^2}{i^2} \frac{q(i^2)}{q((d+1)^2)}\) as follows.
  Recall the expression for \(q(t) = q_{I}(t) = \prod_{i \in [(d+1)^2]^* \setminus I} (t-i)\).
  Hence, we get the following expression for \(\binom{(d+1)^2}{i^2} \frac{q((d+1)^2 - i^2)}{q((d+1)^2)}\):
  \begin{align*}
    \binom{(d+1)^2}{i^2} \frac{q((d+1)^2 - i^2)}{q((d+1)^2)} &=
    \binom{(d+1)^2}{i^2} \frac{\prod_{j \in [(d+1)^2]^* \setminus I} ((d+1)^2 - i^2 - j)}{\prod_{j \in [(d+1)^2]^* \setminus I} ((d+1)^2 - j)} \\
                                                &= \binom{(d+1)^2}{i^2} \frac{\prod_{j \in [(d+1)^2-2, (d+1)^2-3, \ldots, 1]} ((d+1)^2 - i^2 - j)}
                                                {\prod_{j \in [(d+1)^2-2, (d+1)^2-3, \ldots, 1]} ((d+1)^2 - j)}
                                                \\
                                                &= \binom{(d+1)^2}{i^2} \frac{\prod_{j \in [2, 3, 5, \ldots, (d+1)^2-1]} (j - i^2)}
                                                {\prod_{j \in [2, 3, 5, \ldots, (d+1)^2-1]} (j)}
                                                \\
                                                &= \frac{\prod_{j \in [d+1]} j^2}{\prod_{j \in [d+1]^* \setminus \{i\}} (j-i)(j+i)}
                                                = (-1)^i \cdot 2 \cdot \frac{\binom{d+1}{i}}{\binom{d+i+1}{i}}
  \end{align*}
  The calculation here is similar to the calculation performed in Section 6.1 of~\cite{bun2022approximate}, following Fact 32 in the survey.
  Now, this calculation allows us to infer the following upper bound on the absolute value of the final ratio, for any \(i \in [d+1]\):
  \[
    \abs*{\binom{(d+1)^2}{i^2} \frac{q((d+1)^2 - i^2)}{q((d+1)^2)}} \leq 2
  \]
  Hence, we get that the RHS of inequality~\ref{ineq:expand} is at most \(\sum_{i=0}^{d+1} 2 = 2(d+2) \leq 4 \sqrt{n}\) for large enough \(n\).
  Also, with a little rearrangement, the LHS of inequality~\ref{ineq:expand} equals \(\abs*{Z(a) + R_{a, a'} + \frac{f(a)}{2}}\).
  Therefore, if \(\abs*{\ip*{Z+\frac{f}{2},\overline{\gamma}}} = \abs*{Z(a) + R_{a, a'} + \frac{f(a)}{2}} > \frac{1}{\sqrt{n}} \geq \frac{1}{20n} 4 \sqrt{n} \geq \varepsilon \norm{\gamma}_1\) holds, then \(Z\) is good, with \(\overline{\gamma}\) as the witness.
  This completes the proof.
\end{proof}

\subsection{Proposed Directions}

Finally, we pose the question to extend the set of good \(Z\)s.

\begin{problem}
  Find witnesses to extend the set of good \(Z\)s.
\end{problem}

For example, can one find feasible solutions to prove that any \(Z\) with \(\{-1, 0, 1\}\) as its range is good?

\begin{problem}
  Prove that any \(Z : \{0, 1\}^n \to \{-1, 0, 1\}\) is good.
\end{problem}

Such statements will complement Theorem~\ref{thm:zbounds}, as they will prove lower bounds on the values of \(Z\).
We believe that finding more structure in \(\nulls(M(n, d))\) will help toward these problems.
To that effect, we present a simple construction for vectors in \(\nulls(M(n, d))\), which we could not find in current literature.

\begin{construction}
  \label{cons:orthvec}
  Fix some \(n \in \mathbb{N}\), \(d \in \mathbb{N}\) such that \(d < \floor*{\frac{n}{2}}\), and \(k \in \mathbb{N}\) such that \(d < k < n-d\).
  Construct a vector \(\gamma \in \mathbb{R}^{2^n}\), indexed by \(2^{[n]}\), as follows:

  If the indexing set \(T\) has size \(\abs{T} \neq k\), then set \(\gamma_T = 0\).
  Otherwise, for each \(i \in [d+1]\), check whether \(\abs{T \cap \{i, d+1+i\}} = 1\).
  If the check above fails for some \(i \in [d+1]\), set \(\gamma_T = 0\).

  Otherwise, \(T\) intersects exactly once with each pair \(\{i, d+1+i\}\).
  Then, set \(\gamma_T = (-1)^{\abs{T \cap [d+1]}}\).

  Output \(\gamma\).
\end{construction}

We claim that the construction above produces a vector \(\gamma \in \nulls(M(n, d))\).
The proof follows by a simple argument, which we present below.

\begin{lemma}
  \label{lem:orthvec}
  For any \(n \in \mathbb{N}, d \in \mathbb{N}\) and \(k \in \mathbb{N}\), such that \(d < k < n-d\), Construction~\ref{cons:orthvec} produces a vector \(\gamma \in \nulls(M(n, d))\).
\end{lemma}

\begin{proof}
  To show that \(\gamma \in \nulls(M(n, d))\), we need to prove that each of its rows is orthogonal to \(\gamma\).
  Hence, for each row indexed by a set \(S \in \binom{[n]}{\leq d}\), we need to show that \(\sum_{T \in 2^{[n]}} 1_{S \subseteq T} \cdot \gamma_T = \sum_{S \subseteq T} \gamma_T = 0\).
  Note that for sets \(T\) with size \(\abs{T} \neq k\), we have \(\gamma_T = 0\) as per the construction.
  Hence, we can simplify the expression further as \(\sum_{S \subseteq T : \abs{T} = k} \gamma_T = 0\).

  We break the analysis in two cases.

  \begin{enumerate}
    \item \(\abs{S} = d\).
      In this case, we need to consider two possibilities, based on the intersection of \(S\) with \(\{i, i+d+1\}\) pair for each \(i \in [d+1]\).
      \begin{itemize}
        \item Consider the case when \(S\) intersects with some \(\{i_0, i_0+d+1\}\) pair twice.
          Then, any set \(T\) containing \(S\) will also intersect \(\{i_0, i_0+d+1\}\) twice.
          Hence, \(\gamma_T = 0\) for any \(T\) containing such a set \(S\).
          Therefore, \(\sum_{S \subseteq T} \gamma_T = 0\).

        \item Consider the case when \(S\) intersects at most once with each \(\{i, i+d+1\}\) pair for \(i \in [d+1]\).
          Then, as \(\abs{S} = d = k-1\), there is exactly one pair \(\{i_0, i_0+d+1\}\) with which \(S\) does not intersect.

          Consider any set \(T\) containing \(S\), with size \(\abs{T} = k\).
          If \(T\) does not intersect \(\{i_0, i_0+d+1\}\), then \(\gamma_T = 0\) by construction.
          Otherwise, either \(i_0 \in T\), \(i_0+d+1 \in T\), or both.
          If both \(i_0 \in T\) and \(i_0+d+1 \in T\), then \(\gamma_T = 0\) by construction.
          Hence, we only consider the case when either \(i_0 \in T\) or \(i_0+d+1 \in T\), but not both.

          If \(i_0 \in T\), then \(\gamma_T = (-1)^{\abs{T} \cap [d+1]} = (-1)^{\abs{S \cap [d+1]} + 1}\).
          Otherwise, if \(i_0+d+1 \in T\), then \(\gamma_T = (-1)^{\abs{S \cap [d+1]}}\).
          Moreover, both the cases occur for exactly the same number of sets \(T\).
          Hence, \(\sum_{S \subseteq T} \gamma_T = 0\).
      \end{itemize}

    \item
      \label{itm:smallsize}
      \(\abs{S} = s\), where \(s < d\).
      Our goal in this case is to prove \(\sum_{S \subseteq T : \abs{T} = k} \gamma_T = 0\) indirectly using the previous case.
      To that effect, we compare this expression with \(\sum_{S \subseteq S' : \abs{S'} = d} \sum_{S' \subseteq T : \abs{T} = k} \gamma_T\).
      We claim that the latter expression is a multiple of the former.

      To see why, consider a fixed \(T\).
      We can remove any element of \(T \setminus S\) from \(T\) to obtain \(S'\) in \(k - s\) ways.
      Then, the way to obtain \(S\) from \(S'\) is unique, which is to remove all elements of \(S' \setminus S\).
      Hence, we get that
      \[
        (k-s) \cdot \sum_{S \subseteq T : \abs{T} = k} \gamma_T =
        \sum_{S \subseteq S' : \abs{S'} = d} \sum_{S' \subseteq T : \abs{T} = k} \gamma_T
      \]
      Now, from the previous case, we have that \(\sum_{S' \subseteq T : \abs{T} = k} \gamma_T = 0\) for any \(S'\) with \(\abs{S'} = d\).
      Hence, \(\sum_{S \subseteq T: \abs{T} = k} \gamma_T = 0\) as well.
  \end{enumerate}

  The proof is complete.
\end{proof}

Intuitively speaking, the construction above produces a vector with balanced \(-1\) and \(1\) entries on the given Hamming layer.
Hence, it may help with \(Z\)s that are highly unbalanced along the entries of this vector.
We leave it open to characterize the set of \(Z\)s that can be ruled out based on this vector.

\begin{problem}
  Characterize the set of good \(Z\)s witnessed by the vector from Construction~\ref{cons:orthvec}.
\end{problem}

\section{Lower Bounds for \textsf{AND}}
\label{sec:lb}
In this section, we prove a lower bound for torus polynomials approximating \(\AND\).
We restate the result below.

\andlb*

We plan to use Theorem~\ref{thm:asymconvert} to prove this result.
In this task, the main challenge is that there are infinitely many choices for \(Z\), and we need to find a feasible solution \(\gamma\) for each choice.
To make this task easier, we find a vector \(\gamma \in \nulls(M(n, d))\) with integer entries and small \(\ell^1\)-norm.
This vector will serve as a feasible solution for any \(Z\), if it satisfies the conditions that we detail in the next statement.

\begin{lemma}
  \label{lem:shorttolb}
  Fix a Boolean function \(f : \{0, 1\}^n \to \{0, 1\}\) and some \(d \in [n-1]\).
  Consider a vector \(\gamma \in \nulls(M(n, d))\) with integer entries, i.e. \(\gamma \in \mathbb{Z}^{2^n}\), such that \(\ip{f, \gamma}\) is an \emph{odd} integer.
  Then, for any \(Z : \{0, 1\}^n \to \mathbb{Z}\) and \(\varepsilon < \frac{1}{2\norm{\gamma}_1}\), the following holds:
  \[\abs*{\ip*{Z + \frac{f}{2}, \gamma}} > \varepsilon \norm{\gamma}_1\]
\end{lemma}

\begin{proof}
  Let \(\ip{f, \gamma} = 2z + 1\) for some integer \(z \in \mathbb{Z}\).
  Then, \(\ip*{\frac{f}{2}, \gamma} = z + \frac{1}{2}\).
  Fix a function \(Z : \{0, 1\}^n \to \mathbb{Z}\).
  Then, \(\ip*{Z, \gamma} = z'\) for some integer \(z' \in \mathbb{Z}\).
  Hence, \(\ip*{Z + \frac{f}{2}, \gamma} = z' + z + \frac{1}{2}\).
  Consider the following two cases:

  \begin{itemize}
    \item \(z + z' \geq 0\). In this case, \(\ip*{Z + \frac{f}{2}, \gamma} \geq \frac{1}{2}\).
    \item \(z + z' \leq -1\). In this case, \(\ip*{Z + \frac{f}{2}, \gamma} \leq -\frac{1}{2}\).
  \end{itemize}

  In both the cases, we observe that \(\abs*{\ip*{Z + \frac{f}{2}, \gamma}} \geq \frac{1}{2}\).
  Note that by the choice of \(\varepsilon < \frac{1}{2\norm{\gamma}_1}\), we have \(\varepsilon \norm{\gamma}_1 < \frac{1}{2}\).
  This completes the proof.
\end{proof}

To use the preceding statement for \(f = \AND_n\), we proceed as follows.
For a vector \(\gamma \in \nulls(M)\), its inner product with \(\AND_n\) is \(\ip{\gamma, \AND_n} = \gamma_{[n]}\).
Hence, we need a vector \(\gamma \in \nulls(M) \cap \mathbb{Z}^{2^n}\), such that \(\gamma_{[n]}\) is odd and \(\norm{\gamma}_1\) is not too large.
We find this vector in a basis for \(\nulls(M)\), we will find the basis useful later, such that each vector in that basis is integral and has a small enough \(\ell^1\)-norm.
One such vector \(\overline{\gamma}\) in this basis will have \(\abs*{\overline{\gamma}_{[n]}} = 1\).
We present the construction below.

\begin{construction}
  \label{cons:extremebasis}
  Construct a matrix \(B\) as follows.

  \begin{itemize}
    \item[] \textbf{Input}: \(n, d \in [n-1]\).
    \item[] \textbf{Output}: A matrix \(B \in \mathcal{M}_{2^n \times \binom{n}{> d}}(\mathbb{R})\).
    \item[] \textbf{Construction}:
      For any set \(S \subseteq [n]\) with \(\abs{S} \geq d+1\), consider its elements \((s_0, s_1, \ldots, s_{\abs{S}-1})\) in the increasing order \(s_0 < s_1 < \ldots < s_{\abs{S}-1}\).
      Denote \(S[> d] = \{s_{d+1}, \ldots, s_{\abs{S}-1}\}\) as the set remaining after ignoring the first \(d\) elements.

      Now, create a matrix \(B\), of size \(2^n \times \binom{n}{> d}\), and index its rows by elements of \(2^{[n]}\) and columns by elements of \(\binom{[n]}{> d}\).
      For \(1 \leq i \leq 2^n\) and \(1 \leq j \leq \binom{n}{> d}\), set the corresponding entry of \(B\) as \(B_{i, j} = (-1)^{\abs{S_i} + \abs{S_j}} \cdot 1_{S_i \subseteq S_j} \cdot 1_{S_j[> d] \subseteq S_i}\).
      Here, \(S_i\) and \(S_j\) denote the sets indexing the \(i\)\textsuperscript{th} row and the \(j\)\textsuperscript{th} column of \(B\), respectively.

      Output \(B\).
  \end{itemize}
\end{construction}

We claim that, when given \(n\) and \(d\) as input, the construction above produces a basis \(B(n, d)\) for \(\nulls(M(n, d))\).
To prove the claim, we plan to use an extension trick.
Consider the matrix \(D(n)\) of size \(2^n \times 2^n\) where \(D_{i, j} = 1_{S_i \subseteq S_j}\).
This extends the matrix \(M(n, d)\) to a square matrix.
Note that \(D(n)\) is an upper triangular matrix with \(1\)s on the diagonal.
Hence, it has full rank.
Therefore, the following statement follows easily, which we state without proof.

\begin{lemma}
  \label{lem:nullbasis}
  A basis \(B'(n, d)\) for \(\nulls(M(n, d))\) consists of the last \(\binom{n}{> d}\) columns of \(D(n)^{-1}\).
\end{lemma}

Now, we describe \(D(n)^{-1}\) explicitly.
This is simply the M\"{o}bius function over the Boolean hypercube.
We prove its correctness for the sake of completeness.

\begin{lemma}
  \label{lem:dinv}
  For \(1 \leq i, j \leq 2^n\), \(D(n)^{-1}_{i, j} = (-1)^{\abs{S_i} + \abs{S_j}} \cdot 1_{S_i \subseteq S_j}\).
\end{lemma}

\begin{proof}
  \renewcommand{\qedsymbol}{\ensuremath{\blacksquare}}
  Towards a proof, consider \((DD^{-1})_{k, i} = \sum_{j = 1}^{2^n} 1_{S_k \subseteq S_j} (-1)^{\abs{S_i} + \abs{S_j}} 1_{S_j \subseteq S_i}\).
  If \(k = i\), the only non-zero entry on the RHS is when \(S_j = S_i = S_k\), which is \(1\), therefore \((DD^{-1})_{i, i} = 1\).
  When \(k > i\), there is no subset \(S_j\) such that \(S_j \subseteq S_i\) and \(S_k \subseteq S_j\), therefore \((DD^{-1})_{k ,i} = 0\).

  For \(k < i\), consider two cases.
  \begin{itemize}
    \item Let \(S_k \subseteq S_i\) with \(\abs{S_i} = \abs{S_k} + m\) where \(m \geq 1\).
      Then, for any \(0 \leq m' \leq m\), there are \(\binom{m}{m'}\) many sets such that \(\abs{S_j} = \abs{S_k} + m'\) and \(S_k \subseteq S_j \subseteq S_i\).
      Hence, \(DD^{-1}_{k, i} = \sum_{m' = 0}^{m} (-1)^{m'} \binom{m}{m'} = 0\).

    \item Let \(S_k \not\subseteq S_i\).
      Consider any set \(S_j\) such that \(S_k \subseteq S_j\).
      Then, there is \(s \in S_j\) such that \(s \notin S_i\).
      Hence, \(S_j \not\subseteq S_i\).
      Similarly, if \(S_j \subseteq S_i\), then \(S_k \not\subseteq S_j\).
      Therefore, \(DD^{-1}_{k, i} = 0\).
  \end{itemize}

  This completes the proof.
\end{proof}

By combining Lemma~\ref{lem:nullbasis} and Lemma~\ref{lem:dinv}, we get a basis \(B'(n, d)\) for \(\nulls(M(n, d))\).
Now, we prove that \(B(n, d) = B'(n, d)R(n, d)\) for some square matrix \(R\) with full rank.
This implies that \(B(n, d)\) and \(B'(n, d)\) have the same rank.
The proof of this statement is a straightforward calculation, which we omit.

\begin{lemma}
  For any set \(S \subseteq [n]\), define \(S[\leq d] = S \setminus S[> d]\).
  Consider the following matrix \(R(n, d)\) of \(\binom{n}{> d} \times \binom{n}{> d}\) dimensions.
  For any \(\binom{n}{\leq d} + 1 \leq i, j \leq 2^n\), the corresponding entry is \(R(n, d)_{i-\binom{n}{\leq d}, j-\binom{n}{\leq d}} = 1_{S_i \subseteq S_j} \cdot 1_{(S_j[\leq d]) \subseteq S_i}\).
  Then, \(B(n, d) = B'(n, d)R(n, d)\).
\end{lemma}

The statement above allows us to prove the claim that \(B(n, d)\) is a basis for \(\nulls(M(n, d))\).

\begin{claim}
  \label{clm:extremebasis}
  Construction~\ref{cons:extremebasis}, on input \(n, d \in [n-1]\), outputs a basis for \(\nulls(M(n, d))\).
\end{claim}

We use this basis to prove Theorem~\ref{thm:andlb} as follows.

\begin{proof}[Proof of Theorem~\ref{thm:andlb}]
  First, construct \(B\) using Construction~\ref{cons:extremebasis}, with \(n\) and \(d\) as the inputs to the construction.
  Now, denote by \(\overline{\gamma}\) the column of \(B\) indexed by \([n]\).
  Clearly, \(\abs*{\overline{\gamma}_{[n]}} = 1\), hence, \(\ip{\AND, \overline{\gamma}}\) is an odd integer.
  Moreover, \(\norm{\overline{\gamma}}_1 = 2^{d+1}\).
  Therefore, if we apply Lemma~\ref{lem:shorttolb}, we get the following lower bound for \(\varepsilon < \frac{1}{2 \cdot 2^{d+1}}\):
  any torus polynomial that \(\varepsilon\)-approximates the \(\AND\) function must have degree more than \(d\).
  This completes the proof.
\end{proof}

\begin{remark}
Consider the probabilistic polynomial that approximates \(\AND\) over \(\mathbb{F}_2\) with probability \(\varepsilon\), from~\cite{razborov1987lower}, and combine it with~\cite[Lemma 14]{bhrushundi2019torus}.
Then, one gets the following upper bound on \(\tdeg_\varepsilon(\AND)\).

\begin{lemma}[\cite{bhrushundi2019torus, razborov1987lower}]
  \label{lem:andub}
  For any \(\varepsilon > 0\), \(\tdeg_\varepsilon(\AND) \leq \log^2\left(\frac{n}{\varepsilon}\right)\).
\end{lemma}

Hence, the lower bound we have proved in Theorem~\ref{thm:andlb} is only quadratically away from the upper bound for any \(\varepsilon = \frac{1}{n^{\Omega(1)}}\).
\end{remark}

Although we had infinitely many linear programs to work with, we have only used one vector for proving their feasibility.
One could use multiple vectors in \(\nulls(M(n, d))\), such that for each dual corresponding to \(Z \in \mathbb{Z}^{2^n}\), one of these vectors is a feasible solution.
By using multiple vectors, we believe that one can get stronger degree lower bounds, bringing it closer to the upper bound.
We leave this task as an open problem.

\begin{problem}
  Bridge the gap between the lower and upper bound for the \(\AND\) function.
\end{problem}

\subsection{The Very Small Error Case}
\label{subsec:tinyerr}

The literature on polynomial approximations usually focuses on inverse-polynomial error regime.
We study the case where the error is very small, as it allows us to partially answer a question posed in~\cite[Problem 6]{bhrushundi2019torus}.
In~\cite[Problem 6]{bhrushundi2019torus}, the authors ask about the relationship between \(\tdeg_\frac{1}{3}(f)\) and \(\tdeg_\varepsilon(f)\).
In general, one can ask if there is an error-reduction procedure for torus polynomials.
This procedure should take as input \(\varepsilon, \varepsilon' < \varepsilon, \tdeg_\varepsilon(f)\), and output an estimate for \(\tdeg_{\varepsilon'}(f)\).
Ideally, the output should not depend on \(f\), and the estimate should be reasonably close to optimal.

In what follows, we prove that torus polynomials require the same degree to approximate \(\AND\) and \(\maj\) when the error is very small.
Now, if we assume Conjecture~\ref{conj:asymlb}, then the degree required to \(\frac{1}{20n}\)-approximate \(\AND\) is much smaller than the degree required for \(\maj\).
Hence, any error-reduction procedure as described above will produce a suboptimal output if it does not depend on the function being approximated, conditionally answering~\cite[Problem 6]{bhrushundi2019torus} for the error-regime \(\varepsilon < \frac{1}{2^{n+1}}\).
Formally, we prove the following result.

\begin{theorem}
  \label{thm:tightlb}
  Depending on the value of \(\varepsilon\), the following cases hold.

  \begin{itemize}
    \item If \(\varepsilon < \frac{1}{2^{n+1}}\), then the following holds for \(f = \maj\) as well as \(f = \AND\), and infinitely many \(n\).
      Any torus polynomial that \(\varepsilon\)-approximates \(f\) has degree \(n\).

    \item If \(\varepsilon \geq \frac{1}{2^{n+1}}\), then the following holds for any Boolean function \(f : \{0, 1\}^n \to \{0, 1\}\).
      There exists a torus polynomial \(P\) of degree at most \(n-1\) approximating \(f\) within an error of \(\varepsilon\).
      Moreover, if \(f\) is symmetric, then we can take \(P\) to be symmetric as well.
  \end{itemize}
\end{theorem}

\begin{proof}
  To prove the lower bound, first construct the matrix \(B\) using Construction~\ref{cons:extremebasis}, with \(n\) and \(d=n-1\) as the inputs to the construction.
  In this case, there is a single vector \(\overline{\gamma}\) in the basis, with \(\overline{\gamma}_S = (-1)^{\abs{S}}\).

  Now, for \(\maj\), we have \(\ip{\maj, \overline{\gamma}} = \sum_{i > \frac{n}{2}} (-1)^{i} \binom{n}{i}\).
  Consider \(n = 2^t\) for some natural number \(t \geq 2\).
  Then, \(\binom{n}{i}\) is even for all \(1 \leq i \leq n-1\), while \(\binom{n}{n} = 1\) is odd.
  Hence, \(\ip{\maj, \overline{\gamma}}\) is odd.
  Therefore, using Lemme~\ref{lem:shorttolb}, the lower bound holds for \(\varepsilon < \frac{1}{2\norm{\overline{\gamma}}_1} = \frac{1}{2^{n+1}}\).

  For \(\AND\), simply note that \(\abs{\overline{\gamma}_{[n]}} = 1\).
  Hence, we get the lower bound for \(\AND\) as well.

  To prove the upper bound, we continue with the vector \(\overline{\gamma}\) from above.
  Now, consider any function \(f : \{0, 1\}^n \to \{0, 1\}\), and proceed as follows based on the value of \(\ip{f, \overline{\gamma}}\).
  \begin{itemize}
    \item \textbf{Case 1:} \(\ip{f, \overline{\gamma}}\) is even.

      In this case, define the function \(Z\), with \(Z([n]) = \frac{-\ip{f, \overline{\gamma}}}{2}\), and \(Z(a) = 0\) otherwise.
      Then, \(\ip{Z + \frac{f}{2}, \overline{\gamma}} = 0\).

    \item \textbf{Case 2:} \(\ip{f, \overline{\gamma}}\) is odd.
      
      In this case, define the function \(Z\), with \(Z([n]) = -\floor*{\frac{\ip{f, \overline{\gamma}}}{2}}\), and \(Z(a) = 0\) otherwise.
      Then, \(\abs*{\ip{Z + \frac{f}{2}, \overline{\gamma}}} = \frac{1}{2}\).
  \end{itemize}

  In both the case, we have found a \(Z\), such that \(\abs*{\ip{Z + \frac{f}{2}, \overline{\gamma}}} \leq \frac{1}{2} = \frac{1}{2^{n+1}} \norm{\overline{\gamma}}_1\).
  Now, note that \(\overline{\gamma}\) along generates \(\nulls(M(n, n-1))\), hence, all vectors in \(\nulls(M(n, n-1))\) are multiples of \(\overline{\gamma}\).
  Therefore, for any vector \(\gamma \in \nulls(M(n, n-1))\), we have \(\abs*{\ip{Z + \frac{f}{2}, \gamma}} \leq \varepsilon \norm{\gamma}_1\).
  To complete the proof, note that Theorem~\ref{thm:asymconvert} is an equivalence statement.
  This proves the upper bound.

  Finally, if \(f\) is symmetric, then simply note that the \(Z\) we have constructed is also symmetric.
  This suffices to construct a symmetric torus polynomial, we leave the small details to the reader.
\end{proof}

The upper bound above only states the existence of torus polynomials approximating Boolean functions for a tiny error.
We leave it as an open problem to explicitly construct them.

\begin{problem}
  For \(\varepsilon = \frac{1}{2^{n+1}}\) and a given function \(f : \{0, 1\}^n \to \{0, 1\}\), construct a torus polynomial of degree at most \(n-1\) that \(\varepsilon\)-approximates \(f\).
\end{problem}

\section{Lower Bounds for the Symmetric Case}
\label{sec:symlb}
In this section, we prove lower bounds against symmetric torus polynomials approximating symmetric functions.
To describe the general method, we proceed along the lines of the proof for Theorem~\ref{thm:asymconvert}, with two crucial modifications.
The polynomial is symmetric if each monomial of the same degree has the same coefficient, which reduces the number of variables to \(d+1\).
Similarly, as the function is symmetric, the number of constraints reduces to \(n+1\).
To describe the much shorter linear program we obtain, we first construct the following matrix.

\begin{construction}
  \label{cons:symmatm}
  Define the matrix \(\Ms(n, d)\) of size \((d+1) \times (n+1)\).
  For \(0 \leq j \leq d, 0 \leq i \leq n\), the corresponding entry is \(\Ms_{j, i} = \binom{i}{j}\).
\end{construction}

Now, the analogue of Theorem~\ref{thm:asymconvert} for symmetric torus polynomials is as follows.
Note that we define symmetric functions with \([n]^*\) as the domain, such that \(f(i)\) is the output when the input has Hamming weight \(i\).

\begin{theorem}
  \label{thm:symconvert}
  The following statements are equivalent for any \(n, d\), \(\varepsilon < \frac{1}{10}\), and any symmetric function \(f: [n]^* \to \{0, 1\}\).

  \begin{itemize}
    \item Any symmetric torus polynomial that \(\varepsilon\)-approximates \(f\) has degree more than \(d\).
    \item For any function \(Z : [n]^* \to \mathbb{Z}\), there exists a vector \(\psi \in \nulls(\Ms(n, d))\) such that:
      \[\abs*{\ip*{Z + \frac{f}{2}, \psi}} > \varepsilon \norm{\psi}_1\]
  \end{itemize}
\end{theorem}

The proof of this statement is very similar to the proof of Theorem~\ref{thm:asymconvert}, which we skip.
The challenge here is also similar, having to deal with infinitely many linear programs, and proving that they are all feasible.
We proceed similar to the previous section, by choosing a short enough \(\psi\) with integer entries, which gives us the following analogue of Lemma~\ref{lem:shorttolb}.

\begin{lemma}
  \label{lem:symshorttolb}
  Fix a symmetric Boolean function \(f : [n]^* \to \{0, 1\}\) and some \(d \in [n]^*\).
  Consider a vector \(\psi \in \nulls(\Ms(n, d))\) with integer entries, i.e. \(\psi \in \mathbb{Z}^{n+1}\), such that \(\ip{f, \psi}\) is an \emph{odd} integer.
  Then, for any \(Z : [n]^* \to \mathbb{Z}\) and \(\varepsilon < \frac{1}{2\norm{\psi}_1}\), the following holds:
  \[\abs*{\ip*{Z + \frac{f}{2}, \psi}} > \varepsilon \norm{\psi}_1\]
\end{lemma}

Now, we recall the statement of our main lower bound against symmetric torus polynomial approximations.

\andsymlb*

\begin{proof}
  To prove this statement, we claim that for any \(d \leq O\left(\sqrt{\frac{n}{\log(n)}}\right)\), \(\nulls(\Ms(n, d))\) contains a vector \(\overline{\psi}\) with \(\ell^\infty\)-norm \(1\).
  Moreover, the last entry of \(\overline{\psi}\) is \(\overline{\psi}_n = 1\).
  Hence, its inner product with the \(\AND\) function is \(\ip{\AND, \overline{\psi}} = 1\), which is an odd integer.
  We state the claim formally below.

  \begin{claim}
    \label{clm:shortvec}
    For some universal constant \(c\), consider any large enough \(n \in \mathbb{N}\) and \(d \leq \sqrt{\frac{n}{c \log(n)}}\).
    Then, \(\nulls(M(n, d))\) contains a vector \(\overline{\psi}\) with \(\norm{\overline{\psi}}_\infty = 1\) and \(\overline{\psi}_n = 1\).
  \end{claim}

  Assume, for now, that the claim is true.
  Then, we have \(\ip{\AND, \overline{\psi}} = 1\).
  Further, we get \(\norm{\overline{\psi}}_1 \leq (n+1) \norm{\overline{\psi}}_\infty \leq n+1\).
  Finally, we use Lemma~\ref{lem:symshorttolb}, using \(\overline{\psi}\), which we can apply for any \(\varepsilon \leq \frac{1}{2\norm{\overline{\psi}}_1} \leq \frac{1}{2(n+1)}\).
  Note that our choice of \(\varepsilon = \frac{1}{20n}\) satisfies this inequality for large enough \(n\).
  This completes the proof.
\end{proof}

\begin{remark}
  Buhrman, Cleve, de Wolf and Zalka~\cite{buhrman1999bounds} proved an upper bound of \(O\left(\sqrt{n \log\left(\frac{1}{\varepsilon}\right)}\right)\) on the degree of a real polynomial approximating the \(\AND\) function within an error of \(\varepsilon\).
  Note that we can consider this real polynomial, after symmetrizing, as a symmetric torus polynomial approximating the \(\AND\) function within an error of \(\varepsilon\).
  For \(\varepsilon = \frac{1}{O(n)}\), this gives an upper bound of \(O(\sqrt{n \log(n)})\) on the degree.
  Hence, the lower bound of \(\Omega\left(\sqrt{\frac{n}{\log(n)}}\right)\) we have proved above is tight within logarithmic factors in \(n\).
\end{remark}

We also get a separation between symmetric torus polynomials and asymmetric torus polynomials as a corollary of Theorem~\ref{thm:andsymlb}.

\begin{corollary}
  \label{corr:symsep}
  Symmetric torus polynomials are weaker than their asymmetric counterparts.
\end{corollary}

\begin{proof}
  We compare the symmetric torus polynomial lower bound with the upper bound from Lemma~\ref{lem:andub}.
  Using Lemma~\ref{lem:andub}, we get \(\tdeg_\frac{1}{20n}(\AND) \leq \log^2(n)\).
  On the other hand, symmetric torus polynomials require much higher degree to \(\frac{1}{20n}\)-approximate \(\AND\).
  This proves the separation of their power.
\end{proof}

We complete the remaining part of the proof for Theorem~\ref{thm:andsymlb}, which is to prove Claim~\ref{clm:shortvec}.
For this purpose, we will need the following statement from~\cite{bombieri1983siegel}.

\begin{theorem}[\cite{bombieri1983siegel}]
  \label{thm:linfvec}
  Consider a full-rank integer matrix \(B\) of size \(n \times m\), with \(n > m\).
  Then, \(\mathcal{L}(B) = B\mathbb{Z}^m\) contains a non-zero vector \(\mathbf{v} \in \mathbb{Z}^n\) with its \(\ell^\infty\)-norm bounded by:
  \[
    \norm{\mathbf{v}}_\infty \leq \left(\frac{\sqrt{\det(B^TB)}}{D}\right)^{\frac{1}{n - m}}
  \]
  Here, \(D\) is the GCD of all \(m \times m\) minors of \(B\).
\end{theorem}

To prove Claim~\ref{clm:shortvec} using Theorem~\ref{thm:linfvec}, we first describe a basis for \(\nulls(\Ms(n, d))\).
We state the construction without proof, omitting tedious calculations.

\begin{lemma}
  \label{lem:latticebasis}
  For \(n \in \mathbb{N}\), and \(d \in [n-1]^*\), construct a matrix \(\Bs(n, d)\), of size \((n+1) \times (n-d)\).
  Keep the following entry for \(i \in [n]^*\) and \(j \in [n-d-1]^*\): \(\Bs(n, d)_{i, k} = (-1)^{i-k}\binom{d+1}{i-k}\).

  Then, the columns of \(\Bs(n, d)\) form a basis for \(\Ms(n, d)\).
\end{lemma}

Now, to apply Theorem~\ref{thm:linfvec}, we need to estimate \(\det(\Bs(n, d)^T\Bs(n, d))\).
We state the result below, deferring the calculations to the appendix.

\begin{lemma}
  \label{lem:symdet}
  There exists a universal constant \(c\) such that \(\det({\Bs(n, d)}^T\Bs(n, d)) \leq 2^{cd^2 \log(n)}\).
\end{lemma}

Finally, we need the following statement to find a vector \(\psi\) with an odd entry at the desired place.
Informally speaking, consider a vector \(\psi\) such that \(\psi_n = 0\).
Then, we note that \(\Bs\) is a column-circulant matrix, with \(0\)s beyond the two main diagonals.
Hence, we can shift \(\psi\) to obtain \(\overline{\psi}'\) such that \(\overline{\psi}'_n\) contains the \emph{last} non-zero entry of \(\psi\).
We formalize this in the statement below, stated without proof.

\begin{lemma}
  \label{lem:rotation}
  Consider a vector \(\psi \in \nulls(\Ms(n, d))\) such that the maximum index \(i\) where \(\psi_i \neq 0\) is \(i_0\).
  Define the following vector \(\psi'\):
  \[
    \psi'_{i} =
    \begin{cases}
      \psi_{i+i_0-n} & n-i_0 \leq i \leq n\\
      0 & 0 \leq i < n-i_0
    \end{cases}
  \]
  Then, \(\psi' \in \nulls(\Ms(n, d))\).
\end{lemma}

Now, we are ready to prove Claim~\ref{clm:shortvec}, as follows.

\begin{proof}[Proof of Claim~\ref{clm:shortvec}]
  First, choose any \(d \leq \sqrt{\frac{n}{c \log(n)}}\).
  Then, apply Theorem~\ref{thm:linfvec} for \(\nulls(\Ms(n, d))\) with \(\Bs(n, d)\) as its basis.
  Note that \(\Bs(n, d)\) contains a minor, of size \((n-d) \times (n-d)\), with value \(1\).
  Hence, we get \(D = 1\) when we apply Theorem~\ref{thm:linfvec}.
  Therefore, there exists a non-zero vector \(\overline{\psi} \in \nulls(\Ms) \cap \mathbb{Z}^{n+1}\) with \(\norm{\overline{\psi}}_\infty \leq 2^{c \frac{n}{c \log(n)} \log(n) \frac{1}{2(n-d)}} \leq \sqrt{2}\).
  As \(\norm{\overline{\psi}}_\infty\) must be a positive integer, and \(\sqrt{2} < 2\), this shows that \(\norm{\overline{\psi}}_\infty = 1\).

  Now, if \(\overline{\psi}_n = \pm 1\), then, either \(\overline{\psi}\) or \(-\overline{\psi}\) satisfies the claimed conditions.
  Otherwise, if \(\overline{\psi}_n = 0\), use Lemma~\ref{lem:rotation} to obtain \(\overline{\psi'}\) from \(\overline{\psi}\).
  Note that \(\norm{\overline{\psi'}}_1 = \norm{\overline{\psi}}_1 \leq n+1\), and \(\abs{\overline{\psi'}_n} = 1\).
  Hence, we get the desired lower bound using \(\psi'\).
  This completes the proof of the claim.
\end{proof}

\subsection{Error-Degree Trade-off}

The lower bound described in Theorem~\ref{thm:andsymlb} applies to a particular error.
It is natural to attempt to prove a stronger lower bound for a tighter error.
Indeed, we are able to prove stronger lower bounds for tighter errors, but for \(\maj\).
This significantly strengthens the lower bound from~\cite[Corollary 23]{bhrushundi2019torus}.
Following is the statement of the lower bound.

\begin{theorem}
  \label{thm:symtradeoff}
  Fix \(r \in \mathbb{R}\), \(r \geq 0\).
  For any \(\varepsilon \leq 2^{-\Omega(\log^{r+1}(n))}\), any symmetric torus polynomial that \(\varepsilon\)-approximates \(\majn\) has degree at least \(\Omega\left(\sqrt{n\log^r(n)}\right)\).
\end{theorem}

To prove this statement, we plan to use Minkowski's Theorem on the length of shortest vectors in a lattice.
Following is the version we need.

\begin{lemma}[Minkowski's Theorem~\cite{minkowski}]
  \label{corr:lattice}
  Consider a full-rank integer matrix \(B\) of size \(n \times m\), with \(n > m\).
  Then, the lattice \(\mathcal{L}(B) = B\mathbb{Z}^m\) contains a vector \(\mathbf{v} \neq 0\) with
  \[\norm{\mathbf{v}}_1 \leq \sqrt{mn} \mathop{det}(B^TB)^{\frac{1}{2n}}\]
\end{lemma}

We plan to produce short vectors using Minkowski's Theorem, and then invoke Lemma~\ref{lem:symshorttolb} to argue the lower bound.
Toward this, we need to argue that \(\ip{\maj, \overline{\psi}}\) should be odd for a short vector \(\overline{\psi}\) which, unfortunately, we could not prove.
We prove the lower bound indirectly, by looking at a wider class of symmetric functions.
The \(\dw\) function is defined as follows: \(\dw(x) = 1\) if and only if \(\abs{x} = w\).
We obtain the following lower bound for these functions.

\begin{lemma}
  \label{lem:symtradeoff}
  For any large enough \(n \in \mathbb{N}\), there exists a \(w \in [n]^*\), such that the following holds:

  Any symmetric torus polynomial that \(\varepsilon\)-approximates the \(\dw\) function over \(n\) variables, for \(\varepsilon \leq 2^{-\Omega(\log^{r+1}(n))}\), must have degree \(\Omega\left(\sqrt{n\log^r(n)}\right)\).
\end{lemma}

\begin{proof}
  We start by appealing to Minkowski's theorem (Lemma~\ref{corr:lattice}).

  Say \(d = o\left(\sqrt{n\log^r(n)}\right)\) for some \(r \geq 0\), and denote \(\Ms = \Ms(n, d)\).
  We use Lemma~\ref{lem:symdet}, together with Minkowski's theorem, to get that there exists a non-zero vector \(\overline{\psi} \in \nulls(\Ms) \cap \mathbb{Z}^{n+1}\) with \(\norm{\overline{\psi}}_1 \leq n 2^{o\left(\log^{r+1}(n)\right)}\).
  We want to choose a \(w\), and prove the lower bound with respect to that \(\dw\).
  Hence, we need to find a \(w\) such that \(\ip{\overline{\psi}, \dw}\) is odd.

  Consider a non-zero vector \(\overline{\psi} \in \nulls(\Ms) \cap \mathbb{Z}^{n+1}\) with the smallest \(\ell^1\)-norm.
  At least one of the entries of \(\overline{\psi}\) must be odd.
  Otherwise, if they are all even, then \(\overline{\psi}/2\) has strictly smaller \(\ell^1\)-norm.
  The containment \(\overline{\psi}/2 \in \nulls(\Ms) \cap \mathbb{Z}^{n+1}\) follows, because \(\nulls(\Ms)\) is a vector space, and \(\overline{\psi}/2\) has integral entries.
  This is a contradiction.

  This still does not tell us which entry of \(\overline{\psi}\) will be odd.
  Hence, it is still not clear which \(\dw\) we should choose.
  Note that we have to choose \(w\) to write the family of linear programs, all of them must use the same function \(f=\dw\).
  To make our life easy, we turn the problem over its head.

  We notice that we just need to choose \(w\) independent of \(Z\), but it can depend on \(n, d\).
  Hence, as the vector \(\overline{\psi}\) with the smallest \(\ell^1\)-norm depends only on \(n, d\), we can choose \(w\) based on \(\overline{\psi}\).
  This is exactly what we do, we choose \(w\) such that \(\overline{\psi}_w\) is odd.
  Note that at least one such \(w\) must exist, we pick one of them arbitrarily.
  This completes the proof of the statement.
\end{proof}

Now, we apply a trick employed in the proof of~\cite[Corollary 23]{bhrushundi2019torus}.
We describe it below without proof, as it is very similar to~\cite[Lemma 22]{bhrushundi2019torus}.

\begin{lemma}
  \label{lem:majdeltareduction}
  For any \(r \geq 0\), \(\varepsilon \leq 2^{-\Omega\left(\log^{r+1}(n)\right)}\), and large enough \(n\), the following holds:
  If there exists a symmetric torus polynomial of degree \(o\left(\sqrt{n\log^r(n)}\right)\) that \(\varepsilon\)-approximates \(\maj\), then for any \(w \in [n]^*\), there exists a symmetric torus polynomial of degree \(o\left(\sqrt{n\log^r(n)}\right)\) that \(\varepsilon\)-approximates \(\dw\).
\end{lemma}

Now, we can finish the proof of Theorem~\ref{thm:symtradeoff} as follows.

\begin{proof}[Proof of Theorem~\ref{thm:symtradeoff}]
  For some \(r \in \mathbb{R}, r \geq 0\), consider \(\varepsilon = 2^{-\Omega\left(\log^{r+1}(n)\right)}\).
  Assume that there exists a symmetric torus polynomial that \(\varepsilon\)-approximates \(\maj\) with degree \(d = o\left(\sqrt{n\log^r(n)}\right)\).
  Then, for each \(w \in [n]^*\), there exists a symmetric torus polynomial that \(\varepsilon\)-approximates \(\dw\).
  Moreover, each of these polynomials have degree \(d = o\left(\sqrt{n\log^r(n)}\right)\).
  This contradicts Lemma~\ref{lem:symtradeoff}, completing the proof of the theorem.
\end{proof}

To us, it is a bit unsatisfactory that we cannot prove this lower bound for \(\AND\).
Note that the main hurdle for us is not knowing which entry of \(\overline{\psi}\) is odd.
If we can prove that \(\overline{\psi}_n\) is odd, then the lower bound goes through for \(\AND\).
Hence, we state a conjecture that will prove the error-degree trade-off for \(\AND\) as well.
In fact, we believe that the following holds, which is stronger than what we need.

\begin{conjecture}
  For any \(n \in \mathbb{N}, 0 \leq d < n\), there is a vector \(\overline{\psi} \in \nulls(\Ms(n, d)) \cap Z^{n+1}\) with the smallest \(\ell^1\)-norm and \(\overline{\psi}_n = 1\).
\end{conjecture}

\subsection*{Acknowledgements}
We thank Shachar Lovett for helpful feedback on an earlier draft. We also thank anonymous ITCS 2026 referees for their detailed feedback.

\bibliographystyle{alpha}
\bibliography{references}{}

\end{document}